\documentclass{article}
\usepackage{fullpage}
\usepackage{latexsym}
\usepackage[dvips]{graphicx}
\usepackage{url}
\usepackage{amsmath}
\usepackage{amsthm}
\usepackage{pstricks}
\usepackage{pst-grad}
\usepackage{placeins}

\usepackage[T1]{fontenc}
\usepackage[latin1]{inputenc}
\usepackage[american]{babel}

\usepackage{times}
\usepackage{mathptmx}
\usepackage{mathrsfs}

\usepackage{CLRmacros}

\newlength{\savearraycolsep}

\newtheorem{definition}{Definition}
\newtheorem{lemma}{Lemma}
\newtheorem{corollary}{Corollary}

\title{Solving Cyclic Longest Common Subsequence in Quadratic Time}
\author{Andy Nguyen\thanks{Department of Computer Science, Stanford University}}

\begin{document}
\bibliographystyle{acm}
\maketitle

\begin{abstract}
We present a practical algorithm for the cyclic longest common subsequence ($CLCS$) problem that runs in $O(mn)$ time, where $m$ and $n$ are the lengths of the two input strings.  While this is not necessarily an asymptotic improvement over the existing record, it is far simpler to understand and to implement.
\end{abstract}

\section{Introduction}
The longest common subsequence ($LCS$) problem is a classic problem whose myriad extensions are applicable to many fields ranging from string matching for DNA sequence alignment to dynamic time warping for shape distance.  While the original problem considers strings with a beginning and an end as input, a natural extension is to consider cyclic strings, where the starting point is not fixed on either string.  This extension is useful for matching such objects as DNA plasmids and closed curves.  In this paper we present a practical algorithm for the cyclic longest common subsequence ($CLCS$) problem that runs in $O(mn)$ time, where $m$ and $n$ are the lengths of the two input strings.

\section{Related Work}
The state-of-the-art algorithm for regular $LCS$ comes from~\cite{MP80}, who apply the standard dynamic programming algorithm with the ``Four Russians'' speedup~\cite{ADKF70} to achieve a running time of $O(n^2 / \lg n)$.  There are also algorithms (cite) that achieve better runtimes depending on the input, but all of these algorithms are asymptotically slower in the worst case.

The cyclic version of $LCS$ has seen more recent improvements.  The first major improvement comes from~\cite{M90}, which uses a divide-and-conquer strategy to solve the problem in $O(mn \lg m)$ runtime.  Since then, this approach has been applied to generalizations to the problem, and practical optimizations have been discovered; see~\cite{SFC07} for example.  We describe the setup used in these papers in our preliminaries, though we build a different approach on top of this setup to arrive at our new algorithm. \cite{LMS98} provides a solution that runs in $O(n^2)$ time on arbitrary inputs, with asymptotically better performance on similar inputs; however, this solution requires considerable amounts of machinery both to understand and to implement.

\section{Preliminaries}

Given a string $A$ of length $n$ and an integer $k$, $0 \leq k < n$, we can define $cut(A, k)$ = $substring(A, k, n) + substring(A, 0, k)$.  For example, $cut(``\verb|abcd|", 2) = ``\verb|cd|" + ``\verb|ab|" = ``\verb|cdab|"$, while $cut(``\verb|abcd|", 0) = ``\verb|abcd|" + ``" = ``\verb|abcd|"$.  We can extend this definition to any integer $k$ by defining $cut(A, k) = cut(A, k\textrm{ mod } n)$ when $k < 0$ or $k \geq n$.

Let us define the following:

\begin{definition}
Let $A$ and $B$ be strings of length $m$ and $n$ respectively, where $m \leq n$.  A \emph{cyclic longest common subsequence} of $A$ and $B$ $CLCS(A,B)$ is a string that satisfies the following properties:
\begin{enumerate}
\item $CLCS(A,B) = LCS(cut(A,i),cut(B,j))$ for some integers $i$ and $j$.
\item $|CLCS(A,B)| \geq |LCS(cut(A,i),cut(B,j))|$ for all integers $i$ and $j$.
\end{enumerate}
\end{definition}

We can first observe that $CLCS$ is no easier than $LCS$.  To see this, we observe that we can solve $LCS(A,B)$ using $CLCS$ by prepending $A$ and $B$ with $m$ \verb|x|'s, where \verb|x| is a symbol not found in either string.  This forces $CLCS$ to keep $A$ and $B$ aligned with each other to match all the \verb|x|'s, allowing us to extract $LCS(A,B)$ from the result.

On the other hand, we can also show that $CLCS$ is not much harder than $LCS$.  We could solve $CLCS$ by solving $mn$ instances of $LCS$, corresponding to each possible pair of cuts of the strings $A$ and $B$.  This implies that we have a solution to CLCS that runs in $O(m^2 n^2)$ time.  In fact, we can do even better:

\begin{lemma}
There is some cyclic longest common subsequence of $A$ and $B$ $CLCS(A,B)$ that satisfies the following properties:
\begin{enumerate}
\item $CLCS(A,B) = LCS(cut(A,k),B)$ for some integer $k$.
\item $|CLCS(A,B)| \geq |LCS(cut(A,i),cut(B,j))|$ for all integers $i$ and $j$.
\end{enumerate}
\end{lemma}
\begin{proof}
Given in~\cite{M90}.
\end{proof}

This lemma allows us to run only $m$ instances of $LCS$, giving us a runtime of $O(m^2 n)$.  Now let's see if we can do better.  To do so, let's take a short detour to look at regular $LCS$.  Recall that we can solve $LCS(A,B)$ in $O(mn)$ time with the use of a two-dimensional dynamic programming (DP) table that stores lengths $len(i,j)$ and parent pointers $parent(i,j)$.  Normally we look at this table as a two-dimensional array; however, we can also look at this table as a directed grid graph $G_{A,B}$ as follows (Figure~\ref{fig:tabletograph}):

\begin{figure}[htb]
\centering
\includegraphics[width=.95\linewidth]{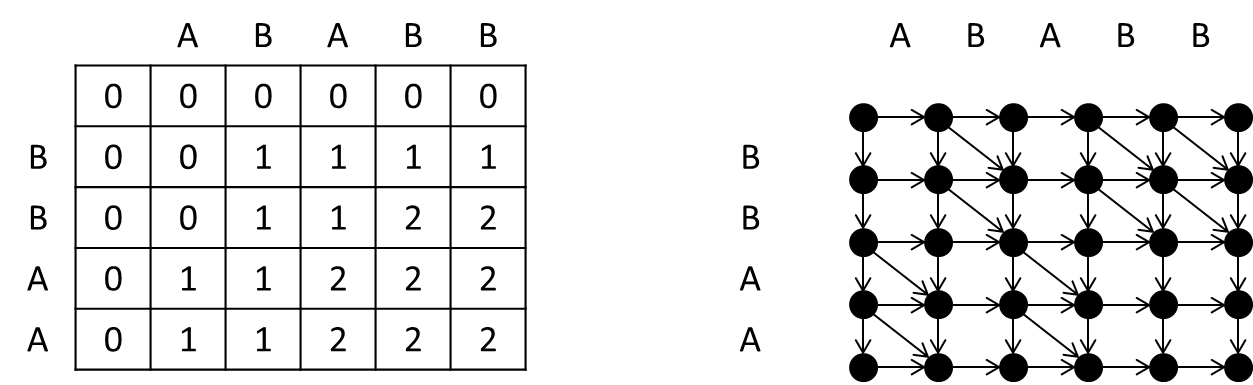}
\caption{Visualizing the DP table as a grid graph.}
\label{fig:tabletograph}
\end{figure}

\begin{itemize}
\item For each entry $e_v$ in the table we have a single node $v$.
\item For each adjacent pair of entries $e_u$ and $e_v$ in the table where $e_u$ is considered when computing $e_v$, we have a directed edge from $u$ to $v$.  Note that all edges point down and/or to the right, so the graph is acyclic.  We see that all of the horizontal and vertical edges are present in the graph, and in addition we have diagonal edges wherever the corresponding characters in the two strings match.
\end{itemize}

For our convenience, we will fix the orientation of the graph as shown in Figure~\ref{fig:tabletograph}, namely, that $(0,0)$ is the top left and $(m, n)$ is the bottom right.  In this orientation, the first coordinate denotes the row, and the second coordinate denotes the column.

\begin{lemma}
Finding an $LCS(A,B)$ is equivalent to finding a shortest path from $(0,0)$ to $(m, n)$ in $G_{A,B}$.
\end{lemma}
\begin{proof}
Given in~\cite{M90}.
\end{proof}

\begin{corollary}
The set of parent pointers computed by any DP solution to $LCS(A,B)$ (when treated as undirected edges) forms a shortest path tree of $G_{A,B}$ rooted at $(0,0)$.
\end{corollary}

Note that this holds true regardless of the choice of parent pointer in case of ties.

This means that we can solve $LCS$ by solving a shortest path problem on a directed acyclic graph with unit edge lengths.

\section{Algorithm}

According to Lemma 2, we can solve $CLCS$ by solving $m$ shortest-path problems, one for each cut of $A$.  These shortest-path problems are closely related, as we see when we construct the graph $G_{AA,B}$ (Figure~\ref{fig:biggraph}).  Then the shortest paths we want to find are the ones from $(0,0)$ to $(m,n)$, from $(1,0)$ to $(m+1,n)$, $\dots$, and from $(m-1,0)$ to $(2m-1,n)$.  Now, we haven't done anything that will save us computation time; it will still cost us $O(m^2 n)$ time to find these $m$ shortest paths.  However, the hope is that since these paths all live on the same graph, we can reuse some work between shortest path computations.

\begin{figure}[htb]
  \centering
  \includegraphics[width=0.5\linewidth]{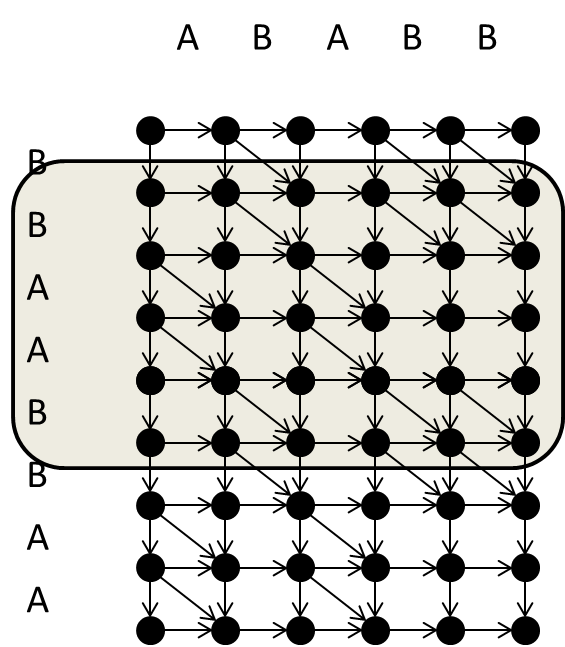}
  \caption{Constructing the graph $G_{AA,B}$.}
  \label{fig:biggraph}
\end{figure}

As we have fixed the orientation of the graph, we can call one edge \emph{lower} than another with respect to this orientation, and we can do so similarly for paths.  We then define the following:

\begin{definition}
A \emph{lowest shortest path tree} of $G_{A,B}$ is a shortest path tree where for every node $v$, the path from the root to $v$ contained in the shortest path tree is lower than any other shortest path from the root to the node.
\end{definition}

First, we show how to compute a lowest shortest path tree efficiently.

\begin{lemma}
The set of parent pointers computed by the DP solution that favors $\leftarrow$  first, then $\nwarrow$, and finally $\uparrow$ in case of ties forms a lowest shortest path tree of $G_{A,B}$ rooted at $(0,0)$.
\end{lemma}
\begin{proof}
Suppose for contradiction that there is some node $v$ where there is a lower shortest path $p'$ from $(0,0)$ to $v$ than the one $p$ contained in the shortest path tree.  Beginning at $v$ and tracing back towards $(0,0)$, $p'$ and $p$ must eventually diverge at some node $w$.  This divergence means that $p'$ and $p$ follow two different valid parent pointers.  But since $p$ follows the parent pointers that favor $\leftarrow$  first, then $\nwarrow$, and finally $\uparrow$, it follows that $p'$ cannot be lower than $p$.
\end{proof}

This lemma implies that maintaining the DP tiebreaking policy locally will yield a lowest shortest path tree globally.  We'll use this to our advantage to bound the work we need to do to extract all $m$ shortest paths.

If we compute the lowest shortest path tree of $G_{AA,B}$, we can immediately read off the shortest path from $(0,0)$ to $(m, n)$.  Before we can read the shortest path from $(1,0)$ to $(m + 1, n)$, we need to re-root the lowest shortest path tree at $(1,0)$.  We can see that this is equivalent to removing the top row entirely, and then fixing the remainder of the tree to restore its properties.

\begin{lemma}
If we remove the top row of nodes from the lowest shortest path tree rooted at $(0,0)$, we are left with at most two subtrees $L$ and (possibly) $R$.
\end{lemma}
\begin{proof}
We show this by showing there are at most two edges between the top row and the rest of the tree.  The vertical edge on the far left, from $(0,0)$ to $(1,0)$, must clearly be in the tree, and we denote the subtree rooted at $(1,0)$ to be $L$.  There cannot be any other vertical edges, since the vertical edge on the far left provides us with a lower path to any node in the second row.  Now consider the first diagonal edge, scanning from left to right, from $(0,k)$ to $(1,k+1)$.  We denote the subtree rooted at $(1,k+1)$ to be $R$.  By similar reasoning, there cannot be any other diagonal edges to the right of this one.
\end{proof}

Note that because of the tiebreaking rule we use for the lowest shortest path tree, any node in $R$ that is adjacent to a node in $L$ (to the left or diagonally up and left; we call these nodes in $R$ \emph{boundary nodes}) chose its parent to be in $R$ because that choice was \emph{strictly} better than the parent option in $L$.

We next note the following about the length entries in the DP table corresponding to $LCS(AA,B)$ (recall that these are lengths of the common subsequence, not lengths of the paths in the graph):
\begin{lemma}
For any entry $(i,j)$ in the table, $0 \leq len(i,j) - len(i,j-1) \leq 1$.
\end{lemma}
\begin{proof}
Clearly $len(i,j) \geq len(i,j-1)$.  Now suppose for contradiction that $len(i,j) - len(i,j-1) > 1$.  We trace back along the parent pointers starting from $(i,j)$ until we enter column $j-1$.  The moment we do so, we find an entry where the length is either $len(i,j)$ if we followed a $\leftarrow$ or $len(i,j)-1$ if we followed a $\nwarrow$.  Either way, this entry is above $(i,j-1)$ in the same column, yet strictly larger, so we could have filled this entry straight down to $(i,j-1)$, which violates the optimality of $len(i,j-1)$.
\end{proof}

This lemma combined with our earlier observation about $L$ and $R$ leads us to the following key fact:
\begin{corollary}
Let $(i,j)$ be such that $(i,j) \in R$ but $(i,j-1) \in L$.  Then $len(i,j) = len(i,j-1) + 1$.
\end{corollary}

Now, if we were to reconnect the two subtrees by setting $parent(1,k+1) = \leftarrow$, we would lower the length values of all nodes in $R$ by exactly $1$ (since we lose the diagonal of $parent(1,k+1)$).  But since every boundary node's old parent choice was strictly better than the option in $L$, lowering the length values of all nodes in $R$ does not result in any inconsistencies in the length table, which means we have a shortest path tree.  This is not necessarily a lowest shortest path tree; however, Corollary 2 tells us that after this reconnection, every boundary node $(i,j) \in R$ satisfies $len(i,j) = len(i,j-1)$, which means in a lowest shortest path tree, $parent(i,j) = \leftarrow$.  Furthermore, every node in $L$ has its length unchanged, and none of its potential parents gained length, so its parent pointer is unchanged.  Similarly, every non-boundary node in $R$ loses exactly one unit of length, as do all of its potential parents, so its parent pointer is also unchanged.

This means if we set the parent of every boundary node to be $\leftarrow$, we now have a lowest shortest path tree, rooted at $(1,0)$.  Note that we can walk along the graph to find all boundary nodes by moving down or right (or both), which means we can perform this rerooting process in $O(m + n) = O(n)$ time.  The pseudocode is given below, where we assume that the input tree is rooted at $(root-1, 0)$.

\codebox{
\<$\procdecl{Re-Root}(root, m, n, parent)$ \\
\li\>  $i = root$ \\
\li\>  $j = 1$ \\
\li\>  \While $j \leq n$ and $parent[i, j] \neq \nwarrow$ \\
\li\>\>  $j = j + 1$ \\
\li\>  \If $j > n$ \\
\li\>\>  \Return \\
\li\>  $parent[i, j] = \leftarrow$ \\
\li\>  \While $i < 2m$ and $j < n$ \\
\li\>\>  \If $parent[i + 1, j]$ == $\uparrow$ \\
\li\>\>\>  $i = i + 1$ \\
\li\>\>\>  $parent[i, j] = \leftarrow$ \\
\li\>\>  \ElseIf $parent[i + 1, j + 1]$ == $\nwarrow$ \\
\li\>\>\>  $i = i + 1$ \\
\li\>\>\>  $j = j + 1$ \\
\li\>\>\>  $parent[i, j] = \leftarrow$ \\
\li\>  \Else \\
\li\>\>\>  $j = j + 1$ \\
\li\>  \While $i < 2m$ and $parent[i + 1, j]$ == $\uparrow$ \\
\li\>\>  $i = i + 1$ \\
\li\>\>  $parent[i, j] = \leftarrow$ \\
}

With this process in place, we can now extract the $m$ shortest paths from $G_{AA,B}$ by computing the lowest shortest path tree of $G_{AA,B}$, and then alternating between reading the shortest path from the root to the appropriate end node and re-rooting the lowest shortest path tree one node lower.  We only need to do so $m$ times, which means the entire process runs in $O(mn)$ time.  The pseudocode is given below.

\codebox{
\<$\procdecl{CLCS}(A, B)$ \\
\li\>  $m = A.length$ \\
\li\>  $n = B.length$ \\
\li\>  let $len[0 \dots 2m, 0 \dots n]$ and $parent[0 \dots 2m, 0 \dots n]$ be new tables \\
\li\>  $\proc{LCS}(AA, B, len, parent)$ \\
\li\>  $S = \proc{Trace-LCS}(m, n, A, parent)$ \\
\li\>  \For $i = 1$ \To $m-1$ \\
\li\>\>  $\proc{Re-Root}(i, m, n, parent)$ \\
\li\>\>  $\hat{S} = \proc{Trace-LCS}(m+i, n, A, parent)$ \\
\li\>\>  \If $S.length < \hat{S}.length$ \\
\li\>\>\>  $S = \hat{S}$ \\
\li\>  \Return $S$ \\
}

\section{Implementation}

If we need to return the subsequence string, we must follow the pseudocode above.  If we only need to return the string's length, we can skip the traceback subroutine and instead simply read off the entry in the length table corresponding to the start of the traceback.  In this case, we need to update the far edge of the length table after each reroot, decrementing the lengths by 1.  In addition to simplifying the code, this change reduces the cache inefficiency of the algorithm, resulting in faster runtimes in practice.

It is worth noting that implementing this algorithm requires far less effort than either the $O(mn \lg m)$ algorithm given in~\cite{M90} or the $O(n^2)$ algorithm given in~\cite{LMS98}.

\section{Conclusions}

We have presented a practical algorithm that solves $CLCS$ in $O(mn)$ time.  Note that the proof of correctness relies heavily on the fact that all edges in the graph have weight 1.  This means that this algorithm does not naturally apply to the usual extensions of the problem, such as edit distance and dynamic time warping.  It may be worth investigating whether the algorithm can be adapted to certain classes of edge weights.

\bibliography{writeup}

\end{document}